\documentclass[11pt]{article}

\usepackage{times}
\usepackage{latexsym}
\usepackage{amsfonts,amsthm,amssymb}
\usepackage{euscript}
\usepackage{amstext}
\usepackage{graphicx}
\usepackage{color}
\usepackage{url,hyperref}
\usepackage{algorithmic}
\newtheorem{lemma}{Lemma}[section]
\newtheorem{theorem}[lemma]{Theorem}

\newtheorem{definition}[lemma]{Definition}
\newtheorem{corollary}[lemma]{Corollary}
\newtheorem{proposition}[lemma]{Proposition}
\newtheorem{claim}[lemma]{Claim}

        {\hspace*{\fill}$\Box$\par}

\newcommand{\opt}{\textrm{\sc OPT}}
\newcommand{\etal}{et al.\ }
\newcommand{\eps}{\epsilon}

\newcommand{\floor}[1]{\lfloor #1 \rfloor}
\newcommand{\ceil}[1]{\lceil #1 \rceil}

\newcommand{\blaps}{\textrm{\sc BLAPS}}
\newcommand{\laps}{\textrm{\sc LAPS}}

\newcommand{\On}{\texttt{On}}

\newcommand{\Opt}{\texttt{Opt}}

\newcommand{\ddphi}{\frac{\mathrm{d}}{\mathrm{dt}} \Phi}
\newcommand{\ddopt}{\frac{\mathrm{d}}{\mathrm{dt}} G^*}
\newcommand{\ddhopt}{\frac{\mathrm{d}}{\mathrm{dt}} G^*}
\newcommand{\ddg}{\frac{\mathrm{d}}{\mathrm{dt}} G}
\newcommand{\ddt}{\frac{\mathrm{d}}{\mathrm{dt}} }
\newcommand{\dt}{\mathrm{dt}}

\newcommand{\cR}{{\cal R}}
\newcommand{\cS}{{\cal S}}
\newcommand{\cQ}{{\cal Q}}

\newcommand{\parr}{\mathcal{P}}

\newcommand{\rank}{\mathbf{rank}}

\begin{document}

\title{Scheduling to Minimize Energy and Flow Time in Broadcast Scheduling }

\author{
Benjamin Moseley\thanks{Department of Computer Science, University of
Illinois, 201 N.\ Goodwin Ave., Urbana, IL 61801. {\tt
bmosele2@illinois.edu}. Partially supported by NSF grant
CNS-0721899.}\\\\\vspace{4mm} 
}

\date{}
\maketitle \vspace{-7mm}

\begin{abstract}

In this paper we initiate the study of minimizing power consumption in the broadcast scheduling model.  In this setting there is a wireless transmitter.  Over time requests arrive at the transmitter for pages of information.  Multiple requests may be for the same page.  When a page is transmitted, all requests for that page receive the transmission simulteneously.  The speed the transmitter sends data at can be dynamically scaled to conserve energy.  We consider the problem of minimizing flow time plus energy, the most popular scheduling metric considered in the standard scheduling model when the scheduler is energy aware.   We will assume that the power consumed is modeled by an \emph{arbitrary} convex function. For this problem there is a $\Omega(n)$ lower bound.  Due to the lower bound, we consider the resource augmentation model of Gupta \etal \cite{GuptaKP10}.  Using resource augmentation, we give a \emph{scalable} algorithm.  Our result also gives a scalable \emph{non-clairvoyant} algorithm for minimizing weighted flow time plus energy in the standard scheduling model.

\end{abstract}

\setcounter{page}{0} \thispagestyle{empty} \clearpage

\section{Introduction}

Wireless transmitters can typically be utilized in a variety of ways.  This  presents the designer with tradeoffs between power, signal range, bandwidth, cost, ect.  In this paper we consider the trade off between power and performance.  Reducing the energy consumption is of importance in many systems.  For example, in ad hoc wireless networks a typical transmitter is battery operated and, therefore, energy conservation is critical.  There are two modes of a wireless transmitter where power can be saved. (1) During idle times (2) During transmit times. Work on the first mode focuses on  putting the system to sleep when not in use \cite{Chen01,Ad00}.  By appropriately putting the system to sleep, energy consumption can be drastically reduced.   Another way the transmitter can reduce power is scaling the speed of the wireless transmission  \cite{IraniSG07}. By using more power a signal can be sent at a faster rate or, to save power, the signal can be sent at a slower rate.    We focus on reducing energy in the second mode. Reducing energy by changing the transmission speed is naturally related to the well studied model of speed scaling in scheduling theory.  However, now that we are in interested in wireless networks, this motivates a generalization of the standard speed scaling model.
 
Companies such as IBM and AMD are producing processors whose speed can be dynamically scaled by the operating system.  Now an operating systems can control the power consumed in the system by scaling the processor speed.  To model this is scheduling theory, it is assumed that there is a power function $P$ where $P(s)$ is the power use when running the processor at speed $s$.  A scheduler now not only chooses the job to schedule, but also the speed used to process the job.  In other words, a scheduler determines a job selection policy and an energy policy.  This is known as the speed scaling scheduling model.  

The standard speed scaling model is similar to our setting, except that we are interested in speed scaling in wireless networks.  Like the speed scaling setting, we assume there  is a power function $P$ where $P(s)$ is the power used when transmitting at speed $s$. To model the wireless network, we will assume a well-studied model known as the broadcast model.  In a broadcast network, there is a single server and $n$ pages of information are stored at the server.  Over time clients send requests for pages to the server.  Multiple clients could be interested in the same page of information.  When the sever broadcasts a page $p$, all outstanding requests for that page are simultaneously satisfied.  This is how the broadcast model differs from the standard scheduling model.  In the standard model, a each request (job) is for a unique page and broadcasting (processing) a page satisfies exactly one request. In the \emph{online} setting, the server is not aware of a request until it arrives to the system.  The broadcast model is motivated by applications in multicast systems and in wireless and LAN networks \cite{Wong88,AcharyaFZ95,AksoyF98,Hall03}.  Along with practical interest, the broadcast model has been studied extensively in scheduling literature \cite{BarnoyBNS98,AksoyF98,AcharyaFZ95,BartalM00,Hall03,GandhiKPS06}.  Similar models have been considered in queuing theory and the stochastic scheduling model \cite{DebS73,Deb84,Weiss79,WeissP81}.

The goal of the scheduler is to satisfy the requests in an order which optimizes a quality of service metric. Naturally, this metric depends on the needs of the system.  When speed scaling is allowed, the server has a dual objective.   One is to minimize the power usage and the other is to optimize the quality of service the clients receive.  The most popular metric in speed scaling literature is minimizing a  linear combination of flow time and total energy \cite{AlbersF07,BansalPS09,ChanELLMP09,GuptaKP10,BansalCP09}.  The flow time\footnote{Flow time is also known as waiting time or response time} of a request is the amount of time the server takes to satisfy the  request.  Let $F$ denote the total flow time of a schedule over all requests and let $E$ denote the energy consumption of the schedule.  The scheduler focuses on minimizing $G= F+E$. This has a natural interpretation.   Say the system is willing to spend one unit of energy to reduce $\rho$   units of flow time.  For example, a system designer may be able to justify spending 1 erg of energy to decrease the flow time by $\rho =10$ micro seconds.  The system designer would then desire a schedule that minimizes $F + 10 E$.  By scaling the units of energy and flow time, we can assume that $\rho=1$.  The main contribution of this work is to initiate the study of energy in the broadcast model. We focus on the the problem of minimizing flow time plus energy in the broadcast setting. Besides practical interest in the model, we believe this problem is of theoretical interest as it is a natural extension of previous work. \\

\noindent {\bf Results:}  In this paper we give an algorithm for minimizing total flow time plus energy in the online broadcast setting where the power function is an \emph{arbitrary} convex function.  There is a $\Omega(n)$ lower bound on the problem of minimizing  flow time in broadcast scheduling when all broadcasts are sent at a fixed rate and energy is not included in the objective \cite{KalyanasundaramPV00}. Since the power function is arbitrary, this lower bound also extends to the problem of minimizing flow time plus energy.  Thus we will resort to resource augmentation \cite{KalyanasundaramP95}.  The resource augmentation model we consider is the same as that introduced recently in \cite{GuptaKP10}.  Here, if our algorithm is given $1+\eps$ speed advantage over the adversary, then our algorithm uses power $P(\gamma)$ when broadcasting at a rate of $(1+\eps)\gamma$.  The rest of the paper will be devoted to proving the following theorem.

\begin{theorem}
\label{thm:main}
There exists an algorithm that with $1+\eps$ resource augmentation is $O(\frac{1}{\eps^3})$-competitive for minimizing total flow time plus energy in broadcast scheduling with an arbitrary power function where $0 <\eps \leq 1$. 
\end{theorem}

Our algorithm is called \emph{scalable} since the minimum amount of resource augmentation is used to be $O(1)$-competitive. Given the strong lower bound, this is the best result that can be shown in the worst case setting up to a constant in the competitive ratio. Broadcast scheduling is a generalization of the standard scheduling model.  The standard model can be interpreted as each request being for a unique page.  In fact, energy plus flow time in the broadcast scheduling strictly generalizes weighted flow time plus energy in the standard speed scaling model.  Thus, our result also gives a scalable algorithm for minimizing energy plus weighted flow time  in standard scheduling setting.  In the standard model, an algorithm is said to be \emph{non-clairvoyant} if it is online and does not know the processing time of a job.    The algorithm we introduce is non-clairvoyant when interpreted in standard scheduling model.

\begin{theorem}\label{thm:non}
There exists a non-clairvoyant algorithm that with $1+\eps$ resource augmentation is $O(\frac{1}{\eps^3})$-competitive  for weighted flow time plus energy in the standard speed scaling model.
\end{theorem}

Before this work, no non-trivial non-clairvoyant algorithm was known for (un)weighted flow time plus energy when the power function can be arbitrary in the standard speed scaling setting. However, it was known that no non-clairvoyant algorithm can be $O(1)$-competitive for this problem \cite{ChanELLMP09}.  Thus, resource augmentation is necessary for an algorithm to be $O(1)$-competitive. \\

\noindent {\bf Relation to Previous work:} In the standard scheduling setting, speed scaling has traditionally used power functions of the form $P(s) = s^\alpha$ for some $\alpha >0$.  These power functions were motivated by the fact that  the power used by CMOS based processors is approximately $s^3$. Recently, modeling the power function as an arbitrary convex was suggested to capture the power consumption of more complex systems \cite{BansalKN09}. As mentioned, in this work, we adopt this new general model.

 There is a large amount of literature on scheduling in the broadcast model.  The most popular scheduling metric considered is minimizing average flow time. Recently, it was shown that there exists an online scheduler that is a scalable algorithm for minimizing the total flow time of a schedule when broadcasts are sent at a fixed rate and energy is not considered in the objective \cite{ImM10, BansalKN09}. This work builds on the ideas and techniques given in \cite{BansalKN09}.  In particular, this paper uses the algorithm of \cite{BansalKN09} with the addition feature of being power aware.  Since the power function considered in this paper is an arbitrary convex function, our work generalizes these results.  \\

\noindent {\bf Previous work on Flow-time in Broadcast Scheduling:}  Minimizing flow time (without energy) in broadcast scheduling where broadcasts are sent at a fixed speed has a rich history beginning with the seminal work of \cite{KalyanasundaramPV00,BartalM00}.  In \cite{KalyanasundaramPV00},  Kalyanasundaram \etal showed that no online deterministic algorithm can be $\Omega(n)$-competitive.  This has since been extended to show that no randomized online algorithm can be $\Omega(\sqrt{n})$-competitive \cite{BansalCKN05}.  Due to these strong lower bounds, most previous work has focused on analyzing algorithms in a resource augmentation model \cite{KalyanasundaramP95}.  In this resource augmentation analysis, the online algorithm is given $s$-speed and is compared to a $1$-speed adversary.  An algorithm is said to be $s$-speed $c$-competitive for some objective function if the algorithm's objective when given $s$ speed is within a $c$ factor of the optimal solution's objective schedule given $1$-speed for every request sequence.    An algorithm that is $(1+\eps)$-speed $O(1)$-competitive is called scalable where $0 < \eps \leq 1$ since it is $O(1)$-competitive when given the minimum advantage over the adversary.

In the offline setting,  the problem was shown to be NP-Hard \cite{ErlebachH02,ChangEGK08}.  The first $O(1)$-speed $O(1)$-approximation was found in \cite{KalyanasundaramPV00}. The best positive result in the offline setting using resource augmentation is a $(1+\eps)$-speed $O(1)$-approximation \cite{BansalCKN05}.    Without resource augmentation, \cite{BansalCKN05} gave a $O(\sqrt{n})$-approximation.  This has since been improved to a $O(\log^2 n/ \log \log n)$ approximation \cite{BansalCS06}.  It is long standing open question whether or not this problem admits a $O(1)$-approximation.

It has been difficult to find competitive online algorithms for broadcast scheduling.  The first online algorithm was given by Edmonds and Pruhs in \cite{EdmondsP03}.  This algorithm was shown to be $(4+\eps)$-speed $O(1)$-competitive algorithm via a reduction to a different scheduling problem known as arbitrary speed up curves.   This reduction takes a $s$-speed $c$-competitive algorithm for the speed up curves setting and converts it into a $(2s)$-speed $c$-competitive algorithm for the broadcast setting.  In \cite{EdmondsP04} Edmonds and Pruhs showed that a natural algorithm Longest-Wait-Fist (LWF) is $6$-speed $O(1)$-competitive via a global charging argument.  The analysis of LWF has been improved to show that the algorithm is $(3.4+\eps)$-speed $O(1)$-competitive \cite{Chekuri09lwf}. 

Later, Edmonds and Pruhs \cite{EdmondsP09} gave an algorithm $\laps$ that is scalable in the arbitrary speed up curves setting.  Using the reduction \cite{EdmondsP03} this gives a $(2+\eps)$-speed $O(1)$-competitive algorithm for the broadcast setting. It was a long standing problem whether or not there existed a scalable online algorithm in the broadcast model.  This question was  resolved by Im and Moseley and Bansal \etal by finding scalable algorithms \cite{ImM10,BansalKN09}. The reduction of \cite{EdmondsP03} was improved by Bansal \etal \cite{BansalKN09} to show that a $s$-speed $c$-competitive algorithm for the speed up curves setting can be converted into a $s(1+\eps)$-speed $O(\frac{c}{\eps})$-competitive algorithm for the broadcast setting.  This shows that $\blaps$, the broadcast version of $\laps$, is a scalable algorithm in the broadcast model.  Besides improving the reduction, in \cite{BansalKN09} the algorithm $\blaps$ was analyzed directly via a potential function analysis.  This analysis has since been extended to show an algorithm that is $(2+\eps)$-speed $O(1)$-competitive for the $\ell_2$-norm of flow time \cite{GuptaIKMP10}.  Recently, a scalable algorithm was found for the $\ell_k$-norms of flow time for all $k>0$ \cite{EdmondsIM10}.  \\

\noindent {\bf Previous work on Speed-Scalaing:}  All of the previous work on broadcast scheduling has assumed that the scheduler broadcasts at some fixed speed. Although speed scaling has not been considered in broadcast scheduling, it has been considered extensively in the standard scheduling model.  As mentioned, the standard scheduling model can be interpreted as each request being for a unique page.  Recall that is the speed scaling setting, a power function $P$ is given where $P(s)$ is the power used when running the processor at speed $s$.    We first consider the traditional model where $P(s)=s^\alpha$ and $\alpha > 1$.     In \cite{PruhsUW08} an efficient algorithm was given for  the problem of minimizing flow time offline subject to a bound on the amount of power consumed.  This algorithm can be extended to find a schedule that minimizes the flow time plus energy in the offline setting.  The problem of minimizing flow time plus energy online was first studied by \cite{AlbersF07}.  \cite{BansalPS09} showed that the algorithm that runs jobs with power proportional to the number of outstanding requests is $4$-competitive for unit work jobs.  They also showed the Highest-Density-First algorithm coupled with this power strategy is $O(\frac{\alpha^2}{\log^2 \alpha})$ competitive for weighted flow plus energy.     This is since been improved in \cite{LamLTW08}  for unweighted flow plus energy to give a $O(\frac{\alpha}{\log \alpha})$-competitive algorithm. In \cite{ChanELLMP09} a $O(\alpha^3)$-competitive \emph{non-clairvoyant} algorithm was given.  Chan \etal in \cite{ChanEP09} gave a $O(\log m)$-competitive algorithm for the speed-up curves setting when there is $m$ processors. 

 Recently, \cite{BansalCP09} introduced the problem of minimizing flow time plus energy with an arbitrary convex power function.  Surprisingly, they were able to give an algorithm that is $3$-competitive in this general setting for flow plus energy.  This resolved a central question on whether an algorithm could be $O(1)$-competitive where the competitive ratio does not depend on $\alpha$.   Gupta \etal gave a scalable algorithm for minimizing weighted flow time plus energy in the case where there are $m$ machines, each machine may have a different power function, and each power function is an arbitrary convex function \cite{GuptaKP10}. As mentioned, in this paper we adopt the assumption that the power function is an arbitrary convex function.\ \\

\section{Preliminaries}

We begin by introducing some notation.  There are $n$ pages stored at the server.  Each page has a size $\sigma_p$.   A request for page $p$ is satisfied if it receives $\sigma_p$ pieces of page $p$ in sequential order.  That is, a request receives the transmission of page $p$ from start to finish in that order. This will be further elaborated on later.   Notice that by this definition, multiple requests can be satisfied by a single broadcast. Let $J_{p,i}$ denote the $i$th request for page $p$. Let $a_{p,i}$ be the time this request arrives to the server.  In the online setting, this is the first time the server is aware of the request. Let $f_{p,i}$ be the time the server satisfies request $J_{p,i}$. The total flow time of a schedule is $F= \sum_p \sum_i (f_{p,i} - a_{p,i})$.  The following definitions will be useful.

\begin{definition}[convex function]
A real valued function $f$ is convex if and only if for any $0  \leq \alpha \leq 1$ and any real valued $x$ and $y$ it is the case that $f(\alpha x + (1-\alpha)y) \leq \alpha f(x) + (1- \alpha)f(y)$ 
\end{definition}

\begin{definition}[concave function]
A real valued function $f$ is concave if and only if for any $0  \leq \alpha \leq 1$ and any real valued $x$ and $y$ it is the case that $f(\alpha x + (1-\alpha)y) \geq \alpha f(x) + (1- \alpha)f(y)$ 
\end{definition}

Using the definition of a concave function, we can easily show the following proposition.

\begin{proposition}
\label{prop:one}
For any real valued concave function $f$ where $f(0)=0$, the following holds
\begin{itemize}
\item For any positive real number $x$, $\frac{x-1}{x} \leq \frac{f(x-1)}{f(x)}$

\item  For any positive real number $x$ and any $\alpha \leq 1$, $\alpha f(x) \leq f(\alpha x)$

\end{itemize}
\end{proposition}

We will be given a power function $P: \mathbb{R} \rightarrow \mathbb{R}$.  The value of $P(s)$ is the power used when the server broadcasts with speed $s$.  In \cite{BansalCP09,GuptaKP10} with a small loss in the competitive ratio, it was assumed that power function $P$ satisfies the following conditions,

\begin{enumerate}
\item At all speeds, $P$ is defined, continuous and differentiable

\item  P(0) = 0

\item P is strictly increasing

\item P is strictly convex

\item P is unbounded

\end{enumerate}

In this paper we adopt these assumptions on $P$.  Notice that the speed could be bounded. If the algorithm is given $(1+\eps)$ resource augmentation and the maximum speed is $s_{max}$ then our algorithm's maximum speed is $(1+\eps)s_{max}$.  Let the function $Q= P^{-1}$.  That is, $Q(y)$ is the maximum speed of the processor with a limit of $y$ on the power used.  Notice that $Q(0)=0$ and that $Q$ is concave.  Let $s(t)$ denote the speed used at time $t$. Let $E = \int_{t=0}^\infty P(s(t))\dt$ be the total energy used.  The goal of the scheduler is to minimize $G= F+E$.

\subsection{Fractional Broadcast Scheduling}
In \emph{integral} broadcast scheduling (the standard model), a request $J_{p,i}$ for page $p$  is satisfied if it receives each of $\sigma_p$ pieces of page $p$ in sequential order.  That is, a page is divided into pieces.  At any time the scheduler decides which piece to broadcast.  A request is satisfied if it receives each of the pieces from start to finish, in that order.   We now define a different version of the broadcast scheduling problem called \emph{fractional} broadcast scheduling.  In this setting, the $\sigma_p$ unit sized pieces of page $p$ are indistinguishable.  Now a request $J_{p,i}$ is satisfied once the server devotes a total of $\sigma_p$ time units to page $p$ after time $a_{p,i}$.  In \cite{BansalKN09} it was shown that an online algorithm running at a fixed speed that satisfies a request $J_{p,i}$ by time $t$ in a fractional schedule can be converted into a different online algorithm that completes $J_{p,i}$ by time $t + \frac{3}{\eps'}(t- a_{p,i}) + \frac{5}{\eps'}$ in an integral broadcast schedule.  Here, the new algorithm is given an additional $\eps'$ resource augmentation.  

It is not obvious that this generalizes when the server can use varying speeds over time. In Appendix \ref{sec:red} we extend their result to the speed scaling setting. We prove the following theorem.  This theorem may be of independent interest, since it can be used to reduce integral broadcast scheduling to the fractional setting for a variety of objective functions where speeds can vary over time. 

\begin{theorem}
\label{thm:red}
Let $S$ be a fractional broadcast schedule where the server can vary its speed over time and let $0 \leq \eps' \leq 1$.  The schedule $\cS$ can be converted into a schedule $\cS'$ using $\eps'$ resource augmentation such that the schedule $\cS'$ satisfies the following properties
\begin{itemize}
\item The power used by $\cS'$ is at most the power used by $\cS$.
\item If a request $J_{p,i}$ is fractionally satisfied at time $f_{p,i}$ under $S$ then this request is integrally satisfied at time $f_{p,i} + \frac{2}{\eps'}(f_{p,i} - a_{p,i}) + \frac{5}{\eps'}$ under the schedule $S'$.
\item If the algorithm that generates $S$ is online then so is the algorithm that generates $S'$.
\end{itemize}

\end{theorem}

This theorem implies the following corollary.
\begin{corollary} \label{cor}
An algorithm with $1+\eps$ resource augmentation that is $c$-competitive for flow time plus energy in fractional broadcast scheduling can be converted into an algorithm that is $\frac{7c}{\eps'}$-competitive for integral broadcast scheduling and uses $(1+\eps)(1+\eps')$ resource augmentation where $0 < \eps' \leq 1$.  
\end{corollary}

Due to the previous corollary, we will focus on fractional broadcast scheduling for the rest of this paper.  It is important to note that if each request is for a unique page, the fractional broadcast setting is equivalent to the standing scheduling setting. This reduction is not needed for the proof of Theorem \ref{thm:non}

\subsection{The Algorithm}

Let $N_a(t)$ contain the unsatisfied requests under our algorithm's schedule at time $t$.  Our algorithm broadcasts at speed $Q(|N_a(t)|)$ at time $t$.   Intuitively, since the flow time of the schedule at time $t$ increases by $|N_a(t)|$, the scheduler can afford to use this much power at time $t$.   Now we define how our algorithm distributes its processing power.  Here the algorithm $\blaps$ is used.  This algorithm was introduced in \cite{EdmondsP09,BansalKN09}.    We will be assuming that our algorithm is given $(1+ 6\eps)$ resource augmentation and this implies the total speed our algorithm uses at time $t$ is $(1+ 6\eps)Q(|N_a(t)|)$ where $\eps \leq \frac{1}{6}$.  The algorithm $\blaps$  takes a parameter $\beta \leq 1$.  We will fix $\beta = \eps$ later.  Let $N'_a(t)$ be the $\beta|N_a(t)|$ requests in $N_a(t)$ with the latest arrival times. At any time $t$, the algorithm equally distributes its processing power amongst the requests in $N'_a(t)$.  That is, for a request $J_{p,i} \in N'_a(t)$ page $p$ is broadcasted at a rate of $x_{p,i}(t) = \frac{(1+6\eps)Q(|N_a(t)|)}{\beta|N_a(t)|}$. Notice that there could be multiple requests for page $p$ in $N'_a(t)$. Let $S_p(t)$ be the set of requests for page $p$ in $N'_a(t)$. A page $p$ is broadcasted at a rate of $\sum_{J_{p,i} \in S_p(t)} x_{p,i}(t)$ at time $t$.  We call $x_{p,i}(t)$ the amount $J_{p,i}$ contributes to how much page $p$ is broadcasted.  Notice that at any time, the algorithm uses total speed $(1+6\eps)Q(|N_a(t))$.

For brevity, throughout this paper will will be using the following simplifying assumption.   We assume that at each time $t$, $\beta|N_a(t)|$ is integral.  A simple extension of the analysis can be made if this is not the case with the following elaboration.   Let $N'_a(t)$ now be the $\ceil{\beta|N_a(t)}$ requests in $N_a(t)$ with latest arrival times. For the $\floor{\beta|N_a(t)|}$ requests with latest arrival times in $N'_a(t)$, $\blaps$ keeps the value of $x$ the same.  Let $J_{p,i}$ be the only other request in $N'_a(t)$. For this request, we set $x_{p,i} =  (\beta|N_a(t)| - \floor{\beta|N_a(t)|})\frac{(1+\eps)Q(|N_a(t)|)}{\beta|N_a(t)|}$.  That is, $J_{p,i}$ is given processing power proportional to $(\beta|N_a(t)| - \floor{\beta|N_a(t)|})$, the amount $J_{p,i}$ `overlaps' the $\beta|N_a(t)|$ latest arriving requests.

\section{Analysis}

We will be using a potential function argument \cite{Edmonds00}.   Let $\ddg(t)$ be the increase in our algorithm's objective at time $t$.  Likewise, let $\ddopt(t)$ be the increase in $\opt$'s objective at time $t$. Notice that $\ddg(t) = 2|N_a(t)|$ because there are $|N_a(t)|$ outstanding requests at time $t$, which increases the flow time of the schedule by $|N_a(t)|$ and the scheduler uses power $|N_a(t)|$ at time $t$.  Let $G = \int^\infty_0 \ddg(t) \dt$ denote our algorithm's total cost and let $G^* = \int^\infty_0 \ddopt(t) \dt$ denote the optimal solution's total cost.  We will define a potential function $\Phi(t)$ that will satisfy the following conditions:

\begin{enumerate}
\item {\bf Boundary Condition}: $\Phi$ is 0 before any request arrives and $0$ after all requests complete
\item {\bf Arrival/Completion Condition}: $\Phi$ does not increase when a request is completed by $\laps$ or $\opt$ or when a request arrives.
\item {\bf Running Condition}:  At all times $t$ it is the case that $\ddg(t) + \ddphi \leq c \ddhopt(t)$ where $c$ is some positive constant.
\end{enumerate}

By integrating the running condition over time, this is sufficient to show that our algorithm is $c$-competitive. This can be seen as follows.  The second equality holds due to $\Phi(0) = \Phi(\infty) = 0$.

$$G = \int^\infty_0 \ddg(t) \dt = \int^\infty_0 (\ddg(t)  + \ddphi(t) )\dt\leq \int^\infty_0 c \ddopt(t)\dt \leq cG^*$$

Now we define our potential function.   We assume without loss of generality that all requests arrive at distinct times, which will simplify the definition of the potential function.  For a request $J_{p,i} \in N_a(t)$ let $\rank(J_{p,i}) = \sum_{J_{q,j} \in N_a(t), a_{q,j} \leq a_{p,i}}1$ at time $t$ be the number of jobs that have arrived during $[0,a_{p,i}]$  that are unsatisfied by the algorithm.  Recall that $x_{p,i}(t)$ is the amount page $p$ is broadcasted by our algorithm at time $t$ due to $J_{p,i}$ and $\sigma_p$ is the amount of page $p$ that must be broadcasted after $a_{p,i}$ to satisfy $J_{p,i}$.  Let $\On(p,i,t_1,t_2) = \int^{t_2}_{t_1} x_{p,i}(t)dt$.   Let $y^*_p(t)$ be the amount of page $p$ that is broadcasted at time $t$ by $\opt$ and let $\Opt(p,t_1,t_2) = \int^{t_2}_{t_1} y_p^*(t)dt$.  We define the variable $z_{p,i}(t)$ for a request $J_{p,i}$ to be $\frac{\On(p,i,t,\infty) \Opt(p,a_{p,i},t)}{\sigma_p}$.  Our potential function is,

$$\Phi(t) := \frac{1}{\eps} \sum_{J_{p,i} \in N_a(t)}   \rank(J_{p,i}) \Big ( \frac{ z_{p,i}}{Q(\rank(J_{p,i}))} \Big)$$

 Our potential function combines the potential functions of \cite{BansalKN09} and \cite{GuptaKP10}. The potential function of \cite{BansalKN09} was used for broadcast scheduling without energy, which needs to somehow have the power function reflected in the potential if it is to work in our setting. To incorporate the power function we use some ideas given in \cite{GuptaKP10}.   

As is usually the case, our potential is designed to approximate the algorithm future cost after subtracting the optimal solution's future cost.    Our algorithm gives higher priority to jobs that have arrived recently.  The potential function is designed to capture the the remaining cost of the algorithm if it satisfies requests in the opposite order of arrival.    Fix a request $J_{p,i}$.  If the optimal solution satisfied request $J_{p,i}$ then the value of $z_{p,i}(t)$ should be thought of as the remaining volume of page $p$ that must be broadcasted to satisfy request $J_{p,i}$.  Assume $J_{p,i}$ has the highest rank in $N_a(t)$. Then $Q(\rank(J_{p,i}))$ is the speed that will be used at time $t$.  The value of $\rank(J_{p,i})$ is the number of requests waiting for request $J_{p,i}$ to be satisfied. Thus,  if requests are satisfied in opposite order of arrival, $\rank(J_{p,i}) \Big ( \frac{ z_{p,i}}{Q(\rank(J_{p,i}))} \Big )$ is an estimate on the flow time that will be accumulated while the algorithm satisfies $J_{p,i}$ because it will take at least $\frac{ z_{p,i}}{Q(\rank(J_{p,i}))}$ time units to satisfy request $J_{p,i}$.

\subsection{Change of the Potential Function}

It is easy to see that the potential function is not affected when the optimal solutions completes a request. Also the potential function does not increase when a request $J_{p,i}$ arrives since $z_{p,i}=0$. Further, the potential is $0$ before all requests arrive and after all requests are completed.  We now show that $\Phi$ does not increase when a request is satisfied by the algorithm.  This lemmas, combined with the previous arguments shows that $\Phi$ satisfies the boundary and arrival/completion conditions.

\begin{lemma}
\label{lem:comp}
$\Phi$ does not increase when the algorithm completes a request.
\end{lemma}
\begin{proof}

Consider a time $t$ where that algorithm completes a request $J_{p,i}$.   We can assume that $\rank(J_{p,i}) < |N_a(t)|$; otherwise, trivially there is no increase in $\Phi$.   The change in $\Phi$ is,

$$\Delta \Phi(t)  = \frac{1}{\eps}\sum_{J_{p',j} \in N_a(t), a_{p',j} > a_{p,i}}^{|N_a(t)|}  \Big ( \frac{(\rank(J_{p',j}) - 1)z_i}{Q(\rank(J_{p',j})-1)} - \frac{\rank(J_{p',j}) z_i}{Q(\rank(J_{p',j}))} \Big )  $$

To show that $\Delta \Phi(t) \leq 0$, we need to show that 

\begin{eqnarray*}
\sum_{J_{p',j} \in N_a(t), a_{p',j} > a_{p,i}}^{|N_a(t)|} \frac{(\rank(J_{p',j}) - 1)z_i}{Q(\rank(J_{p',j})-1)}  &\leq& \sum_{J_{p',j} \in N_a(t), a_{p',j} > a_{p,i}}^{|N_a(t)|}  \frac{\rank(J_{p',j}) z_i}{Q(i)}   \Rightarrow \\
\sum_{J_{p',j} \in N_a(t), a_{p',j} > a_{p,i}}^{|N_a(t)|} \frac{(\rank(J_{p',j}) - 1)}{\rank(J_{p',j})}  &\leq& \sum_{J_{p',j} \in N_a(t), a_{p',j} > a_{p,i}}^{|N_a(t)|}  \frac{Q(\rank(J_{p',j})-1)}{Q(\rank(J_{p',j}))} 
\end{eqnarray*}

This is true by definition of $Q$ and Proposition \ref{prop:one}
\end{proof}

Now we concentrate on the running condition, the final property of $\Phi$ that needs to be shown.   Fix a time $t$.  Let $N_o(t)$ be the set of requests unsatisfied by $\opt$ at time $t$.  First lets consider the change in $\Phi(t)$ due to the algorithm's processing.   Recall that the algorithm broadcasts page $p$ at a rate of $\frac{(1+6\eps)Q(|N_a(t)|)}{\beta |N_a(t)|}$ due to a request $J_{p,i} \in N'_a(t)$.  Further, notice that $\Opt(p,a_{p,i}, t) \geq \sigma_p$ for any request $J_{p,i} \in N_a(t) \setminus N_o(t)$, since $\opt$ must broadcast $\sigma_p$ units of page $p$ after $a_{p,i}$ to satisfy page $p$.  Therefore, for any $J_{p,i} \in N'_a(t) \setminus N_o(t)$ we have $\ddt z_{p,i}(t) \geq \ddt \On(p,i,t,\infty) =  -\frac{(1+6\eps)Q(|N_a(t)|)}{\beta |N_a(t)|} $. Notice that the $\rank$ of each request the algorithm  is currently working on is at least $(1-\beta)N_a(t)$.  This is because $N'_a(t)$ consists of the $\beta|N_a(t)|$ requests with latest arrival times.  Using this, we can upperbound the change in $\Phi(t)$ due to the algorithm's processing,

\begin{eqnarray}
\ddphi &\leq& - \frac{1}{\eps}\sum_{J_{p,i} \in N'_a(t) \setminus N_o(t) } \rank(J_{p,i}) \frac{(1+6\eps)Q(|N_a(t)|)}{\beta N_a(t) Q(|N_a(t)|)}  \leq  -\sum_{J_{p,i} \in N'_a(t) \setminus N_o(t) } \frac{(1+6\eps)(1-\beta)}{\eps\beta}  \label{eqn:algproc}
\end{eqnarray}

Next, we consider the change in $\Phi$ due to the adversary's processing.  Let $p^*$ be the page which the adversary broadcasts.  Let $s_o(t)$ be the speed the optimal solution processes page $p^*$ at.  Let $P^*(t) = P(s_o(t))$ be the power $\opt$ uses at time $t$.    The adversary can affect $z_{p^*,i}$ for any request $J_{p^*,i} \in N_a(t)$.  We first observe the following,

$$\ddt \sum_{ J_{p^*,i} \in N_a(t)} z_{p^*,i}(t)  \leq \sum_{ J_{p^*,i} \in N_a(t)} \frac{\On(p^*,i,t,\infty)}{\sigma_{p^*}} \cdot \Big(\ddt \Opt(p^*,a_{p^*,i},t)  \Big) \leq \ddt \Opt(p^*,a_{p^*,i},t) = s_o(t) \dt $$

The last inequality holds since $\sum_{i, J_{p^*,i} \in N_a(t)} \On(p^*,i,\infty) \leq \sigma_{p^*}$.  This is because the algorithm needs to only broadcast page $p^*$ for a total of $\sigma_{p^*}$ amount of time to satisfy all outstanding requests for page $p^*$. Let $J_{p^*,k}$ be the request in $N_a(t)$ such that $\frac{\rank(J_{p^*,k})}{Q(\rank(J_{p^*,k}))}$ is maximized.  We can upper bound the increase in $\Phi$ due to $\opt$' processing as,

$$\ddphi = \frac{1}{\eps} \Big ( \frac{\rank(J_{p^*,k})}{Q(\rank(J_{p^*,k}))} \Big) \ddt \sum_{ J_{p^*,i} \in N_a(t)}z_{p^*,i}(t) \leq  \frac{\rank(J_{p^*,k})s_o(t)}{\eps Q(\rank(J_{p^*,k}))}$$

Our goal is to show that $\ddg(t) + \ddphi(t) \leq \frac{2}{\eps^2} \ddopt(t)$.  First we consider the case when the adversary uses power at least $|N_a(t)|$.  In this case, the increase in the algorithm's objective can be charged directly to the optimal solution, along with any increase in the potential function.

\begin{lemma}
\label{lem:optpower}
If $Q(|N_a(t)|) \leq  s_o(t)$ (equivalently, $P^*(t) \geq |N_a(t)|$) then $\ddg(t) + \ddphi(t) \leq \frac{2}{\eps} \ddhopt(t)$.
\end{lemma}
\begin{proof}
First we bound the increase in $\Phi(t)$ due to $\opt$'s processing. Intuitively, the increase in $\opt$'s objective, due to using a lot of power, is large enough to absorb the increase in $\Phi$ and  the increase in algorithm's  objective.  By increasing $\Phi$ now and charging it to $\opt$, we can later use the decrease in $\Phi$ to pay for increases in the algorithm's objective. Recall that $J_{p^*,k}$ is request in $N_a(t)$ for page $p^*$ that maximizes $\frac{\rank(J_{p^*,k})}{Q(\rank(J_{p^*,k}))}$.  Let $\alpha |N_a(t)| = \rank(J_{p^*,k})$ and let $\gamma  = \frac{s_o(t)}{Q(|N_a(t)|)}$.    Notice that $\alpha \leq 1$ and $\gamma \geq 1$.    The function $Q$ is concave.  Therefore, $Q(\rank(J_{p^*,k})) \geq \alpha Q(|N_a(t)|)$ by Proposition \ref{prop:one}.  The increase in $\Phi(t)$ due to $\opt$'s processing can be bounded as follows.
   
\begin{eqnarray*}
\Big(\frac{\rank(J_{p^*,k})}{\eps} \Big ) \frac{s_o(t)}{Q(\rank(J_{p^*,k}))} \leq \frac{1}{\eps}\alpha |N_a(t)| \gamma \frac{Q(|N_a(t))|}{Q(\rank(J_{p^*,k}))} \leq \frac{1}{\eps}\alpha |N_a(t)| \gamma \frac{Q(|N_a(t))|}{\alpha Q(|N_a(t)|)}  \leq  \frac{\gamma}{\eps} |N_a(t)| 
\end{eqnarray*}

The total power used by $\opt$ is at least $P^* \geq \gamma |N_a(t)|$ since $P$ is convex and the speed $\opt$ uses is $\gamma Q(|N_a(t)|)$. Hence, $\ddhopt(t) \geq \gamma|N_a(t)|$. Recall that $\ddg(t)= 2|N_a(t)|$.  Knowing that $\eps \leq \frac{1}{6}$ we have that,

\begin{eqnarray*}
\ddg(t) + \ddphi(t)  \leq 2 |N_a(t)| + \frac{\gamma}{\eps} |N_a(t)| \leq \frac{2 \gamma}{\eps} |N_a(t)| \leq \frac{2}{\eps}\ddopt(t)
\end{eqnarray*}

\end{proof}

Due to the previous lemma, for the rest of this paper we can concentrate on the case where $Q(|N_a(t)|) >  s_o(t)$.  We begin by bounding the increase in $\Phi(t)$ due to $\opt$'s processing using this assumption,

\begin{eqnarray}
\Big ( \frac{\rank(J_{p^*,k})}{\eps} \Big ) \frac{s_o(t)}{Q(\rank(J_{p^*,k}))} &\leq& \Big ( \frac{\rank(J_{p^*,k})}{\eps} \Big ) \frac{Q(|N_a(t)|)}{Q(\rank(J_{p^*,k}))} \nonumber \\
 &\leq&   \Big ( \frac{\rank(J_{p^*,k})}{\eps} \Big ) \frac{(|N_a(t)|/\rank(J_{p^*,k}))Q(\rank(J_{p^*,k}))}{Q(\rank(J_{p^*,k}))} \nonumber \\
 &\leq& \frac{1}{\eps} |N_a(t)| \label{eqn:optproc}
\end{eqnarray}

The first inequality holds since we assumed that $Q(|N_a(t)|) >  s_o(t)$.  The second inequality follows from definition of $Q$ and Proposition \ref{prop:one}. We can now prove the final case of the running condition.

\begin{lemma}\label{lem:main}
If $s_o(t) < Q(|N_a(t)|) $ and $\beta = \eps$ then $\ddg(t) + \ddphi(t) \leq \frac{2}{\eps^2} \ddhopt(t)$.
\end{lemma}
\begin{proof}
We know that $\ddg(t) = 2|N_a(t)|$.  The increase in $\Phi(t)$ due to $\opt$'s processing is at most $\frac{1}{\eps}|N_a(t)|$ and the change in $\Phi(t)$ due to the algorithm's processing is at most $-\sum_{J_{p,i} \in N'_a(t) \setminus N_o(t) } \frac{(1+6\eps)(1-\beta)}{\eps\beta}$.  Combining these, we have

\begin{eqnarray*}
\ddg(t) + \ddphi(t) &\leq& 2|N_a(t)| + \frac{1}{\eps}|N_a(t)|-\sum_{J_{p,i} \in N'_a(t) \setminus N_o(t) } \frac{(1+6\eps)(1-\beta)}{\eps\beta} \\
&\leq&  2|N_a(t)| + \frac{1}{\eps}|N_a(t)|- \Big (\sum_{J_{p,i} \in N'_a(t)  } \frac{(1+6\eps)(1-\beta)}{\eps\beta} - \sum_{J_{p,i} \in N_o(t)  } \frac{(1+6\eps)(1-\beta)}{\eps\beta} \Big ) \\
&\leq& 2|N_a(t)| + \frac{1}{\eps}|N_a(t)| -  \Big ( \frac{(1+6\eps)(1-\beta)}{\eps} |N_a(t)| -  \frac{(1+6\eps)(1-\beta)}{\eps\beta} |N_o(t)| \Big ) \\
&\leq& \frac{2}{\eps^2} |N_o(t)| \leq  \frac{2}{\eps^2} \ddhopt(t)
\end{eqnarray*}

The second to last inequality follows from $\beta = \eps \leq \frac{1}{6}$.  The last inequality follows since the optimal solution's flow time increases by $|N_o(t)|$ at time $t$.
\end{proof}

Combing Lemmas (\ref{lem:optpower}) and (\ref{lem:main}) along with setting $\beta$ to be $\eps$ gives the running condition, namely $\ddg(t) + \ddphi(t) \leq \frac{2}{\eps^2} \ddhopt(t)$. Thus, we have shown all the properties of $\Phi$, which gives the following theorem.

\begin{theorem}
The algorithm with $\beta = \eps \leq \frac{1}{6}$ is $(1+6\eps)$-speed $\frac{2}{\eps^2}$-competitive for flow plus energy in fractional broadcast scheduling.
\end{theorem}

By using  the reduction from integral to fractional broadcast scheduling in Corollary (\ref{cor}) and scaling $\eps$, we have proven Theorem \ref{thm:main}.

\section{Conclusion}

In this paper we initiated the study of energy in the broadcast scheduling model.  We showed a scalable algorithm for average flow time plus energy.  It is important to note that the algorithm $\blaps$ explicitly depends on the speed $\eps$. That is, $\beta$ depends on the speed $\eps$. Recently,  many algorithms developed for scheduling have this dependency \cite{ChekuriIM09,EdmondsP09,GuptaIKMP10}.  It would be interesting to find an algorithm which does not depend on $\eps$ or show no such algorithm exists.  It would also be of interest to consider objective functions that include energy minimization for the broadcast setting.

\bibliographystyle{plain}
\bibliography{Energybroad}

\appendix

\section{Reduction of Integral to Fractional Broadcast Scheduling}\label{sec:red}

In this section we prove Theorem \ref{thm:red}. Consider any sequence of requests and let $\cS$ denote a valid fractional broadcast schedule.   We now define an algorithm to construct an integral schedule $\cS'$.  In the integral schedule we will use $0 <\eps \leq 1$ resource augmentation over the schedule $\cS$.  Recall that in an integral schedule, a request $J_{p,i}$ must receive $\sigma_p$ unit sized pieces of page $p$ in \emph{sequential} order.  We will assume that the fractional schedule works on at most one page during a unit time slot.  We can assume this without loss of generality because we can set a unit time slot to be arbitrarily small, since we are assuming preemption. Let $f_{p,i}$ denote the time request $J_{p,i}$ is fractionally satisfied in $\cS$.  Let $f^I_{p,i}$ denote the time that $J_{p,i}$ is integrally satisfied in $\cS'$. Our algorithm and analysis build on the reduction from fractional to integral broadcast scheduling given in  \cite{BansalKN09} where broadcasts are always sent at a fixed speed. Since the processor speeds can vary in our setting, in our analysis we will have to be careful about how speed is distributed and how processing power is accounted for.

Our algorithm will keep track of a queue $\cQ$ of tuples.  A tuple will be of the form $\langle p, w, b, k \rangle$.   Here $p$ corresponds to a page. The value of $k \in \mathbb{R}$ will correspond to the part of page $p$ that will be broadcasted.  The variable $w \in \mathbb{R}$ will be called the \emph{width} and $b \in \mathbb{R}$ will be the \emph{start time}.     Each tuple $\tau = \langle  p, w, b, k \rangle$ will correspond to some request $J_{p,i}$ and the width will be $w= f_{p,i} - a_{p,i} $.  This request will determine the part of page $p$ that will be broadcasted and the speed of the broadcast.  The request corresponding to a tuple could be updated.  However, after the first piece of the page is broadcasted, the request corresponding to the tuple will not be changed. 

At any time when not broadcasting a page, our algorithm will determine a tuple $\tau= \langle  p, w, b, k \rangle$ to broadcast.  Let $J_{p,i}$ correspond to this tuple.  Let $s(p,i,k)$ denote the speed used by  $\cS$ in the $k$th time slot devoted to page $p$ after time $a_{p,i}$.  Notice that a $s(p,i,k)$ volume of page $p$ is broadcasted in this time slot under $\cS$.   When broadcasting the tuple $\tau$, we mean that our algorithm is broadcasting the page $p$ with speed $(1+\eps)s(p,i,k)$.  Our algorithm stops broadcasting $p$ once a $s(p,i,k)$ volume of work is completed.  Notice that this takes $\frac{1}{1+\eps}$ time because our algorithm runs at speed $(1+\eps)s(p,i,k)$ and $\cS$ took one time unit to broadcast the $s(p,i,k)$ volume of page $p$.  If $k= 1$ then the algorithm broadcasts $p$ from the beginning of the page.  Otherwise, our algorithm broadcasts $p$ from where it left off.    Besides scaling the speeds, the algorithm is the same as that used in  \cite{BansalKN09}.

\begin{center}
\begin{tabular}[r]{|c|}
\hline
\textbf{Algorithm}: GenRounding(t) \\

\begin{minipage}{14cm}
\begin{algorithmic}
\STATE All requests arrive unmarked
\STATE Simulate the schedule $\cS$
\FOR{Any unmarked request $J_{p,i}$  completed by $\cS$ at time $t$}
\IF{There is a tuple $\tau=\langle p, w, b, k \rangle \in \cQ$ for page $p$ where $b \geq a_{p,i}$}
\STATE Update the width of $\tau$ to be $w = \min\{w, f_{p,i} - a_{p,i} \}$
\ELSE
\STATE Insert the tuple $\langle p, (f_{p,i} - a_{p,i}), \infty, 1 \rangle$ into $\cQ$
\ENDIF
\ENDFOR
\STATE Dequeue the tuple $\tau=\langle p, w, b, k \rangle $ with minimum width, breaking ties arbitrarily
\STATE Let $J_{q,j}$ be the request which gave $\tau$ the width $w$
\STATE Broadcast the tuple  $\tau$ 
\IF{This broadcast was the first unit piece of page $p$. That is, $k=1$}
\STATE Set $b = t$
\ENDIF
\IF{This broadcast was the last unit piece of page $p$. That is, $\sum_{i=1}^k s(q,j,k) = \sigma_p$}
\STATE Mark all requests for page $p$ that arrived before time $s$
\ELSE
\STATE Update $\tau$ to be $\langle p, w, b, k+1 \rangle $ 
\ENDIF
\end{algorithmic}
\end{minipage}\\\\

\hline
\end{tabular}
\end{center}

It can be observed that our algorithm is online if $\cS$ is online. First we show that the power used by our algorithm is most the power used by $\cS$.

\begin{lemma}\label{lem:overlap}
Consider any two tuples $\tau_j$ and $\tau_k$ for the same page  $q$ broadcast by $\cS'$. Let $J_{q,j}$ and $J_{q,k}$ be the requests which gave these tuples their width and speed, respectively.  Let $j>k$ then $a_{q,j} \geq f_{q,k}$    
\end{lemma}
\begin{proof}
The algorithm $\cS'$ put the tuple $\tau_k$ into $\cQ$ at time $f_{q,k}$.  Hence, the start time $b(\tau_k)$ of $\tau_k$ must be after $f_{q,k}$.   Since the algorithm inserted a new tuple for $J_{q,j}$, it must be that $ b(\tau_k) < a_{q,j}$. Hence,   $f_{q,k} \leq b(\tau_k) < a_{q,j}$.
\end{proof}

Consider any unit time slot in $\cS$.  The previous lemma implies that only one tuple broadcasted in $\cS'$ will correspond to this time slot.  This implies the following claim.

\begin{claim}\label{clm:power}
The power used by $\cS'$ is at most the power used by $\cS$.
\end{claim}

It remains to show that $f^I_{p,i} - a_{p,i} \leq \frac{2}{\eps}(f_{p,i} - a_{p,i})+ \frac{5}{\eps}$ for any request $J_{p,i}$. Fix a request $J_{p,i}$.  If at time $f_{p,i}$ the request $J_{p,i}$ is marked then $f^I_{p,i} - a_{p,i} \leq  \frac{2}{\eps}(f_{p,i} - a_{p,i})+ \frac{5}{\eps}$ and we are done. Thus, we assume that $J_{p,i}$ is unmarked at time $f_{p,i}$.  Since, $J_{p,i}$ is unmarked, there exists a tuple $\tau = \langle p, w, b, k \rangle$ in $\cQ$ at time $t$ where $w \leq f_{p,i}  - a_{p,i}$.  Let $t_A$ be the earliest time before time $t$ where every tuple  dequeued during $[t_A,f_{p,i}]$ has width less than $w$.  Let $t_B$ be the latest time after time $t$ where every tuple dequeued during $[f_{p,i}, t_B]$ has width less than $w$. By definition of the algorithm, at each time during $[t_A,t_B]$ only tuples with width at most  $w$ are dequeued.  Let $\prod$ denote the set of tuples broadcasted during $[t_A,t_B]$.  Notice by definition of the algorithm,  once a tuple is broadcasted, the width of the tuple never changes.

\begin{claim}\label{clm1}
Consider any $\tau \in \prod$ and let $J_{q,j}$ be the request that gave $\tau$ its width.  Then $a_{q,j} \geq t_A - w$.
\end{claim}
\begin{proof}
Since the tuple $\tau$ was broadcasted during $[t_A,t_B]$ the width $w(\tau)$ must be less than $w$.   By definition of $t_A$ there are no tuples with width less than $w$ in $\cQ$ at time $t_A-1$.  Thus, $f_{q,j} \geq t_A$ and by definition of width, $a_{q,j}  \geq f_{q,j} - w \geq t_A - w$.
\end{proof}

Let $N = |\prod|$.  All $N$  tuples in $\prod$ correspond to some unit time slot where $\cS$ preformed a broadcast.  By the previous claim, all of these broadcasts occurred after time $t_A - w$. Further,  by definition of the algorithm all of these broadcasts occurred before time $t_B$.  This is because, a tuple corresponding to a request $J_{q,k}$ will not be inserted into $\cQ$ until after time $f_{q,k}$ and all of the requests corresponding to tuples in $\prod$ are satisfied in $\cS'$ by time $t_B$.  This implies that the schedule $\cS$ uses at least $N$ time slots during $[t_A - w, t_B]$. Thus we have,  $$N \leq t_B - t_A + w+2$$

Since our algorithm has $\eps$ resource augmentation, it takes $\frac{1}{1+\eps}$ time to broadcast one tuple.  Since our algorithm is always busy during $[t_A,t_B]$ we have that $$ N \geq (1+\eps)(t_B -t_A) -3$$

This implies that $t_B - t_A \leq \frac{w +5}{\eps}$.  For the request $J_{p,i}$ we know that $a_{p,i} \geq t_A - w$ by Claim \ref{clm1}. Thus, $t_B - a_{p,i} + w \leq \frac{w +5}{\eps}$. Knowing that $f^I_{p,i} \leq t_B$ and $\eps \leq 1$, we have $(f^I_{p,i} - a_{p,i}) \leq  \frac{2w +5}{\eps} $.  Combining this with Claim \ref{clm:power} and the fact that our algorithm is online, we have Theorem \ref{thm:red}.

\end{document}